\documentclass{article}

\usepackage[margin=1in]{geometry}
\usepackage{amsmath,amssymb,amsthm,mathtools}
\usepackage{enumitem}
\usepackage{hyperref}
\usepackage{xcolor}
\usepackage{algorithm}
\usepackage{algpseudocode}

\hypersetup{
  colorlinks=true,
  linkcolor=blue,
  citecolor=blue,
  urlcolor=blue
}

\newtheorem{theorem}{Theorem}
\newtheorem{lemma}{Lemma}

\newtheorem{remark}{Remark}
\newtheorem{definition}{Definition}

\newcommand{\OPT}{\operatorname{OPT}}
\newcommand{\R}{\mathbb{R}}

\title{Directed Reachability-Preserving Minimum Edge Cut: Approximation and Planar Hardness}
\author{Qi Duan, Carnegie Mellon University}
\date{}

\begin{document}

\maketitle

\begin{abstract}
We study a directed version of the three-terminal
reachability-preserving minimum edge cut problem.  Given a directed graph
$G=(V,A)$ with arc costs and terminals $s_1,s_2,t$, the one-way directed
RPMEC problem asks for a minimum-cost set of arcs whose deletion preserves
the reachability $s_1\leadsto s_2$ while destroying the reachability
$s_1\leadsto t$.  We first give a path--cut formulation in terms of a rooted
directed cut function.  Using a root-linear approximation for the associated
polymatroid, we obtain an $O(\sqrt r)$-approximation, where $r$ is the number
of relevant vertices with positive singleton cut value.  In particular this
gives an $O(\sqrt n)$-approximation in general directed graphs.  For acyclic
directed graphs, we give an additional singleton-length algorithm and obtain
an $O(\min\{\sqrt r,h\})$ guarantee, where $h$ is the maximum number of
relevant vertices on an $s_1$-$s_2$ path.  Finally, we prove that directed
planar RPMEC is NP-hard, even on acyclic planar digraphs with nonnegative
costs, by reducing from independent set on cubic planar graphs through a
finite-bimodal directed node-cut construction and a planar node-to-edge split.
\end{abstract}

\section{Introduction}

Connectivity cut problems usually ask for a minimum-cost deletion set that
separates prescribed terminal pairs.  In many applications, however, one also
needs to preserve some desired reachability while blocking an undesired one.
For example, in attack-graph mitigation, one may want to block paths to a
dangerous goal while preserving reachability among benign or required actions.
This motivates reachability-preserving minimum cut problems.

We consider the three-terminal directed edge-deletion problem.  Given
terminals $s_1,s_2,t$, the goal is to delete arcs so that $s_1$ can still
reach $s_2$, but $s_1$ can no longer reach $t$.  We call this the one-way
directed three-terminal RPMEC problem.  This is the directed analogue of the
undirected reachability-preserving minimum edge cut problem, but direction
creates new difficulties.  In particular, preserving reachability is not the
same as keeping two terminals in the same connected component.

The first part of the paper gives approximation results.  We show that the
problem can be written as

\[
\OPT=\min_{P:s_1\leadsto s_2} f(V(P)),
\]

where $f(X)$ is the minimum directed out-cut separating every vertex of $X$
from $t$.  This converts the problem into choosing a protected directed path
under a monotone submodular rooted cut objective.  Applying a root-linear
polymatroid approximation gives an $O(\sqrt r)$ approximation, where $r$ is
the number of relevant vertices with positive singleton cut value.  For DAGs,
a simpler singleton-length algorithm gives an $h$-approximation, where $h$ is
the maximum number of relevant vertices on an $s_1$-$s_2$ path.  Combining the
two gives $O(\min\{\sqrt r,h\})$.

The second part proves hardness for planar directed graphs.  Although
undirected planar variants can be more tractable, the directed planar problem
is NP-hard.  We first build an acyclic planar directed node-deletion instance
whose finite-cost deletable vertices are locally bimodal.  Then we use a
planar node-to-edge split that is valid precisely because of bimodality.  This
yields NP-hardness of directed planar RPMEC, even when the constructed
directed graph is acyclic.

\section{Problem Definitions}

A directed graph is written as $G=(V,A)$, where $V$ is the vertex set and
$A$ is the arc set.  For $U\subseteq V$, define the directed out-cut

\[
\delta^+(U)=\{(u,v)\in A:u\in U,\ v\notin U\}.
\]

For vertices $x,y\in V$, we write

\[
x\leadsto y
\]

if there exists a directed path from $x$ to $y$.

\begin{definition}[One-way directed RPMEC]
The one-way directed three-terminal reachability-preserving minimum edge cut
problem, abbreviated directed RPMEC, is defined as follows.  The input is a
directed graph $G=(V,A)$, nonnegative arc costs
$c:A\rightarrow \R_{\ge 0}$, and terminals $s_1,s_2,t\in V$.  A feasible
solution is an arc set $F\subseteq A$ such that, in $G-F$,

\[
s_1\leadsto s_2
\]

is preserved, while

\[
s_1\not\leadsto t
\]

holds.  The objective is to minimize $c(F)=\sum_{e\in F}c(e)$.
\end{definition}

We also use a node-deletion version in the hardness proof.

\begin{definition}[One-way directed RPMNC]
The one-way directed three-terminal reachability-preserving minimum node cut
problem, abbreviated directed RPMNC, has the same reachability requirements as
directed RPMEC, but the deletion set consists of nonterminal vertices.  Each
nonterminal vertex $v$ has a nonnegative weight $w(v)$, and terminals are
undeletable.
\end{definition}

If $s_1$ does not reach $s_2$ in the input graph, the instance is infeasible.
Throughout the approximation section we assume that at least one directed
$s_1$-$s_2$ path exists.

\section{Directed Three-Terminal RPMEC}
\label{sec:directed-rpmec}

Let

\[
R=V\setminus\{t\}.
\]

For every $X\subseteq R$, define the rooted directed cut function

\[
f(X)
=
\min\left\{
c(\delta^+(U)):
X\subseteq U\subseteq R
\right\}.
\]

Equivalently, $f(X)$ is the minimum cost of an arc set whose deletion
separates every vertex in $X$ from $t$.  This value can be computed by a
directed minimum cut: add a super-source $\sigma$ with infinite-capacity arcs
$\sigma\to x$ for all $x\in X$, and compute a minimum $\sigma$-$t$ cut
\cite{ahuja1993network}.

\subsection{Path--Cut Formulation}

\begin{lemma}[Directed path--cut formulation]
\label{lem:directed-path-cut}
For the one-way directed three-terminal RPMEC problem,

\[
\OPT
=
\min\left\{
f(V(P)):
P \text{ is a directed } s_1\text{-}s_2 \text{ path in } G-t
\right\}.
\]
\end{lemma}

\begin{proof}
Let $F$ be any feasible solution.  Since $s_1\leadsto s_2$ is preserved in
$G-F$, there exists a directed $s_1$-$s_2$ path $P$ in $G-F$.  Since
$s_1\not\leadsto t$ in $G-F$, this path cannot contain $t$.

Let $U$ be the set of vertices reachable from $s_1$ in $G-F$.  Then
$s_1,s_2\in U$, $t\notin U$, and $V(P)\subseteq U$.  Moreover, every arc of
$\delta^+(U)$ must belong to $F$; otherwise the head of such an arc would
also be reachable from $s_1$ in $G-F$, contradicting the definition of $U$.
Therefore

\[
c(F)\ge c(\delta^+(U))\ge f(V(P)).
\]

Taking the minimum over all feasible solutions $F$ gives

\[
\OPT
\ge
\min_{P:s_1\leadsto s_2} f(V(P)).
\]

Conversely, let $P$ be any directed $s_1$-$s_2$ path in $G-t$, and let
$U\subseteq R$ attain $f(V(P))$.  Since $V(P)\subseteq U$, every arc of $P$
has both endpoints in $U$.  Hence deleting $\delta^+(U)$ preserves the path
$P$.  Since $t\notin U$ and all arcs leaving $U$ are deleted, $s_1$ cannot
reach $t$ after the deletion.  Thus $\delta^+(U)$ is a feasible solution of
cost $f(V(P))$.  This proves the reverse inequality.
\end{proof}

The function $f$ is normalized, monotone, and submodular.  Normalization
follows from $f(\emptyset)=0$.  Monotonicity follows because the feasible
family for $Y$ is contained in the feasible family for $X$ whenever
$X\subseteq Y$.  For submodularity, let $U_X$ and $U_Y$ be optimal sets for
$X$ and $Y$, respectively.  The directed out-cut function is submodular:

\[
c(\delta^+(U_X))+c(\delta^+(U_Y))
\ge
c(\delta^+(U_X\cap U_Y))+c(\delta^+(U_X\cup U_Y)).
\]

Since $U_X\cap U_Y$ is feasible for $X\cap Y$ and $U_X\cup U_Y$ is feasible
for $X\cup Y$, we obtain

\[
f(X)+f(Y)\ge f(X\cap Y)+f(X\cup Y).
\]

Thus $f$ is a monotone submodular cut function
\cite{fujishige2005submodular,schrijver2003combinatorial}.

\subsection{General Directed Graphs}

For each $v\in R$, define the singleton directed cut value

\[
\lambda_v=f(\{v\}).
\]

Let

\[
R^+=\{v\in R:\lambda_v>0\},
\qquad
r=|R^+|.
\]

Vertices outside $R^+$ have zero cost to separate from $t$.  They do not
increase the rooted cut value: if $\lambda_v=0$, then adding $v$ to any
terminal set does not increase $f$.  Indeed, by subadditivity and monotonicity,

\[
f(X)\le f(X\cup\{v\})\le f(X)+f(\{v\})=f(X).
\]

Therefore

\[
f(X)=f(X\cap R^+)
\qquad
\text{for all }X\subseteq R.
\]

We use the following standard root-linear approximation theorem for
monotone submodular polymatroid rank functions.  It follows from the
ellipsoidal approximation framework for polymatroids and submodular functions
\cite{goemans2009submodular}.  This theorem is the only external
approximation black box used in this section.

\begin{lemma}[Root-linear approximation of directed rooted cuts]
\label{lem:directed-root-linear}
There are nonnegative vertex lengths $a_v$, $v\in R^+$, computable in
polynomial time up to arbitrary precision using a directed min-cut separation
oracle, such that for every $X\subseteq R^+$,

\[
\sqrt{\sum_{v\in X}a_v}
\le
f(X)
\le
C\sqrt r\sqrt{\sum_{v\in X}a_v},
\]

where $C>0$ is a universal constant.
\end{lemma}

For completeness, we spell out the separation oracle used for this
polymatroid.  Let

\[
P_f=
\left\{
y\in\R_{\ge 0}^{R^+}:
y(X)\le f(X)\ \text{for all }X\subseteq R^+
\right\}.
\]

Equivalently,

\[
P_f=
\left\{
y\in\R_{\ge 0}^{R^+}:
y(W\cap R^+)\le c(\delta^+(W))\ \text{for all }W\subseteq R
\right\}.
\]

The second description is important: the minimizing cut set $W$ may contain
zero-singleton vertices from $R\setminus R^+$, so it is not enough to test
only sets $W\subseteq R^+$.

For a candidate vector $y$, a violated inequality is found by maximizing

\[
y(W\cap R^+)-c(\delta^+(W))
\]

over all $W\subseteq R$.  This can be reduced to a directed minimum cut by
adding a source $\sigma$, adding arcs $\sigma\to v$ of capacity $y_v$ for
$v\in R^+$, and keeping the original arc capacities.  For a source-side set
$\{\sigma\}\cup W$, the cut capacity is

\[
y(R^+\setminus W)+c(\delta^+(W)).
\]

Thus minimizing this cut is equivalent to maximizing

\[
y(W\cap R^+)-c(\delta^+(W)).
\]

\begin{algorithm}[t]
\caption{General Directed 3-Terminal RPMEC Approximation}
\label{alg:directed-rpmec}
\begin{algorithmic}[1]
\Require Directed graph $G=(V,A)$, costs $c$, terminals $s_1,s_2,t$
\State Compute the rooted directed cut oracle $f$
\State Compute $R^+=\{v\in V\setminus\{t\}: f(\{v\})>0\}$
\State Compute root-linear lengths $a_v$ from Lemma~\ref{lem:directed-root-linear}
\State Find a shortest directed $s_1$-$s_2$ path $P$ in $G-t$ using vertex lengths $a_v$
\State Compute a minimum directed cut separating $V(P)$ from $t$, namely $f(V(P))$
\State \Return the corresponding arc cut
\end{algorithmic}
\end{algorithm}

\begin{theorem}[General directed approximation]
\label{thm:general-directed-rpmec}
The one-way directed three-terminal RPMEC problem admits a polynomial-time
$O(\sqrt r)$-approximation, where

\[
r=
|\{v\in V\setminus\{t\}: f(\{v\})>0\}|.
\]

In particular, since $r\le |V|-1$, it admits an $O(\sqrt{|V|})$-approximation.
\end{theorem}

\begin{proof}
Let $P^\star$ be an optimal witness path in the path--cut formulation of
Lemma~\ref{lem:directed-path-cut}.  Thus

\[
f(V(P^\star))=\OPT.
\]

Let $P$ be the shortest directed $s_1$-$s_2$ path in $G-t$ with respect to
the vertex lengths $a_v$.  By Lemma~\ref{lem:directed-root-linear},

\[
\sqrt{\sum_{v\in V(P^\star)\cap R^+}a_v}
\le
f(V(P^\star)\cap R^+)
=
f(V(P^\star))
=
\OPT.
\]

Since $P$ is shortest with respect to $a_v$,

\[
\sum_{v\in V(P)\cap R^+}a_v
\le
\sum_{v\in V(P^\star)\cap R^+}a_v.
\]

Again by Lemma~\ref{lem:directed-root-linear},

\[
f(V(P))
=
f(V(P)\cap R^+)
\le
C\sqrt r\sqrt{\sum_{v\in V(P)\cap R^+}a_v}.
\]

Combining the inequalities gives

\[
f(V(P))
\le
C\sqrt r\cdot \OPT.
\]

The algorithm outputs a minimum directed cut separating $V(P)$ from $t$, whose
cost is exactly $f(V(P))$.  Therefore the approximation ratio is
$O(\sqrt r)$.
\end{proof}

\subsection{Acyclic Directed Graphs}

We now specialize to the case where $G$ is acyclic.  The general
$O(\sqrt r)$ approximation remains valid.  Acyclicity gives an additional
path-length-based guarantee.

Let

\[
h
=
\max\left\{
|V(P)\cap R^+|:
P \text{ is a directed } s_1\text{-}s_2 \text{ path in } G-t
\right\}.
\]

Since $G$ is acyclic, $h$ can be computed by dynamic programming over a
topological ordering.

Consider the following singleton-length algorithm.  For each $v\in R$, compute

\[
\lambda_v=f(\{v\}).
\]

Then find a shortest directed $s_1$-$s_2$ path $P$ in $G-t$ using vertex
lengths $\lambda_v$, and finally output a minimum directed cut separating
$V(P)$ from $t$.

\begin{lemma}[Singleton-length bound on DAGs]
\label{lem:dag-singleton-bound}
On acyclic directed graphs, the singleton-length algorithm is an
$h$-approximation.
\end{lemma}

\begin{proof}
Let $P^\star$ be an optimal witness path.  Then

\[
f(V(P^\star))=\OPT.
\]

For every $v\in V(P^\star)$, monotonicity gives

\[
\lambda_v=f(\{v\})
\le
f(V(P^\star))
=
\OPT.
\]

Since $P$ is shortest with respect to the singleton lengths $\lambda_v$,

\[
\sum_{v\in V(P)}\lambda_v
\le
\sum_{v\in V(P^\star)}\lambda_v.
\]

By subadditivity of $f$,

\[
f(V(P))
\le
\sum_{v\in V(P)} f(\{v\})
=
\sum_{v\in V(P)}\lambda_v.
\]

Therefore,

\[
f(V(P))
\le
\sum_{v\in V(P^\star)}\lambda_v
\le
|V(P^\star)\cap R^+|\cdot \OPT
\le
h\cdot \OPT.
\]

Thus the singleton-length algorithm is an $h$-approximation.
\end{proof}

\begin{theorem}[Acyclic directed approximation]
\label{thm:dag-directed-rpmec}
The one-way directed three-terminal RPMEC problem on acyclic directed graphs
admits a polynomial-time approximation ratio

\[
O\left(\min\{\sqrt r,h\}\right),
\]

where

\[
r=
|\{v\in V\setminus\{t\}: f(\{v\})>0\}|
\]

and

\[
h
=
\max\left\{
|V(P)\cap R^+|:
P \text{ is a directed } s_1\text{-}s_2 \text{ path in } G-t
\right\}.
\]
\end{theorem}

\begin{proof}
Run both algorithms.  Algorithm A is the root-linear algorithm of
Theorem~\ref{thm:general-directed-rpmec}, and Algorithm B is the
singleton-length algorithm of Lemma~\ref{lem:dag-singleton-bound}.  Algorithm
A returns a solution of cost at most $O(\sqrt r)\cdot \OPT$, while Algorithm
B returns a solution of cost at most $h\cdot \OPT$.  Taking the cheaper of
the two returned cuts gives approximation ratio

\[
O\left(\min\{\sqrt r,h\}\right).
\]
\end{proof}

This result shows that acyclicity can improve the approximation guarantee when
the relevant directed depth is small.  For example, if every relevant
$s_1$-$s_2$ path contains at most $O(\log |V|)$ vertices with positive
connectivity to $t$, then the ratio becomes $O(\log |V|)$.  If the relevant
depth is constant, the ratio becomes constant.

Acyclicity alone does not appear to remove the main worst-case difficulty.
A minimum-label path instance can be embedded into the acyclic directed RPMEC
objective.  Given a labeled acyclic graph, associate labels with vertices.
For each label $\ell$, add a label vertex $q_\ell$ and an arc $q_\ell\to t$
with cost equal to the cost of label $\ell$.  For every vertex $v$ with label
$\ell$, add a sufficiently large-cost arc $v\to q_\ell$.  Place all original
vertices first in a topological ordering, then the label vertices, and
finally $t$; the resulting graph is still acyclic.  For any chosen
$s_1$-$s_2$ path $P$, the cheapest way to separate $V(P)$ from $t$ is to pay
once for every distinct label appearing on $P$.  Thus the acyclic directed
RPMEC objective contains a minimum-label path type subproblem.  Minimum-label
path is known to have nontrivial approximation hardness
\cite{hassin2007labeled}.  Therefore, beating the $\sqrt r$ guarantee in the
worst case likely requires additional structure beyond acyclicity alone.

\section{NP-Hardness of Directed Planar RPMEC}
\label{sec:directed-planar-rpmec-hardness}

We now prove that the one-way directed planar edge version of RPMEC is
NP-hard.  The proof has two parts.  First, we construct an acyclic planar
directed RPMNC instance whose finite-cost deletable vertices are all bimodal.
Second, we transfer this node-deletion instance to an edge-deletion instance
by a planar split operation.

\subsection{A Planar Directed Node-Cut Source}

We reduce from \textsc{Independent Set} on cubic planar graphs.  Given a
cubic planar graph $H=(U,E)$ and an integer $K$, the question is whether
there exists an independent set $S\subseteq U$ with

\[
|S|\ge K.
\]

This problem is NP-complete \cite{madhavan1984independent}.  Let

\[
n=|U|.
\]

We construct a directed planar RPMNC instance with terminals

\[
s_1=a,\qquad s_2=b,\qquad t=z
\]

and budget

\[
B=n-K.
\]

The intended interpretation is that deleting the selector vertex $o_u$ means
that the vertex $u$ is not chosen in the independent set.

\subsubsection{The planar tour/lane embedding}

We need a planar way to arrange one serial certificate gadget for each edge
of $H$.  The following elementary topological lemma supplies such an
arrangement.

\begin{lemma}[Planar tour/lane embedding]
\label{lem:planar-tour-lane}
Let $H$ be a plane graph.  In polynomial size, one can construct a planar
layout consisting of:
\begin{enumerate}
    \item a simple directed serial chain from $a$ to $b$;
    \item at least one certificate location on the chain for every edge
    $e\in E(H)$;
    \item noncrossing lanes from each certificate location for
    $e=uv$ to small neighborhoods of the endpoint vertices $u$ and $v$;
    \item a directed planar drain tree from all vertex neighborhoods to a
    common sink $z$.
\end{enumerate}
Moreover, these structures can be drawn without crossings and with mutually
disjoint local lanes.
\end{lemma}

\begin{proof}
Take a regular neighborhood $N(H)$ of a fixed planar embedding of $H$: each
vertex of $H$ becomes a small disk and each edge becomes a thin rectangular
corridor joining the corresponding disks.  By taking the corridors
sufficiently wide, we may reserve several disjoint lanes inside every
corridor.

If $H$ is disconnected, add auxiliary connecting corridors in the outer face
only for the purpose of the layout.  These corridors carry no certificate
gadget and do not represent edges of $H$.  Thus the layout surface can be
assumed connected.

Cut $N(H)$ along a set of noncrossing arcs so that the resulting surface is a
disk.  Equivalently, one may cut along a spanning tree of the planar dual.
The boundary of the resulting disk gives a simple nonselfcrossing boundary
tour.  This tour visits every original edge corridor at least once.  Place the
directed serial chain from $a$ to $b$ along this boundary tour, and place a
certificate location whenever the tour visits the corridor corresponding to an
edge of $H$.

For a certificate location belonging to an edge $e=uv$, route one local lane
inside the corridor of $e$ to the disk of $u$, and another local lane inside
the same corridor to the disk of $v$.  Since distinct original corridors are
interior-disjoint, and since finitely many parallel lanes can be placed inside
each corridor, these routes can be chosen without crossings.

Finally, route a directed tree from all vertex disks to a common sink $z$
along a separate reserved lane.  Since the drain structure is a tree, it can
be embedded without self-crossing.  Placing the serial chain, the leakage
lanes, and the drain tree in disjoint lanes gives the desired planar layout.
\end{proof}

If the boundary tour visits an original edge $e$ more than once, we create
one certificate copy for each visit.  Repeating a certificate for the same
edge does not change the correctness of the reduction.

\subsubsection{The directed RPMNC construction}

For every vertex $u\in U$, create a selector vertex $o_u$ with node weight

\[
w(o_u)=1.
\]

For every certificate copy corresponding to an edge $e=uv$, create a directed
two-branch gadget $D_e^i$ with entry $p_e^i$ and exit $q_e^i$:

\[
p_e^i\to g_{u,e}^i\to q_e^i,
\]

and

\[
p_e^i\to g_{v,e}^i\to q_e^i.
\]

The vertices $g_{u,e}^i$ and $g_{v,e}^i$ are called branch gates.  They have
weight

\[
w(g_{u,e}^i)=w(g_{v,e}^i)=0.
\]

Connect all certificate copies in the order in which they appear on the
planar tour:

\[
a\to p_{e_1}^{i_1},
\qquad
q_{e_j}^{i_j}\to p_{e_{j+1}}^{i_{j+1}},
\qquad
q_{e_L}^{i_L}\to b.
\]

For every branch gate, add a leakage arc to the corresponding selector:

\[
g_{u,e}^i\to o_u,
\qquad
g_{v,e}^i\to o_v.
\]

For every selector $o_u$, add a directed drain path from $o_u$ to $z$ using
the planar drain tree from Lemma~\ref{lem:planar-tour-lane}.  Thus, whenever
$o_u$ is not deleted, reachability of $o_u$ implies reachability of $z$.

All auxiliary vertices other than selectors and branch gates receive weight

\[
M=B+1.
\]

Therefore no feasible node-deletion solution of cost at most $B$ can delete
an auxiliary vertex of weight $M$.

The construction is planar by Lemma~\ref{lem:planar-tour-lane}.  It is also
acyclic: the certificate gadgets form one directed chain from $a$ to $b$,
leakage arcs go from branch gates to selectors, and drain arcs go from
selectors toward the sink $z$.  There are no arcs from the selector/drain
part back to the serial chain.

\begin{lemma}[Independent set implies feasible RPMNC solution]
\label{lem:is-to-rpmnc}
If $H$ has an independent set of size at least $K$, then the constructed
directed planar RPMNC instance has a feasible node-deletion set of cost at
most $B=n-K$.
\end{lemma}

\begin{proof}
Let $S\subseteq U$ be an independent set with $|S|\ge K$.  Delete the
selectors corresponding to vertices not in $S$:

\[
X_{\mathrm{sel}}=\{o_u:u\notin S\}.
\]

The cost is

\[
w(X_{\mathrm{sel}})=|U\setminus S|\le n-K=B.
\]

Now consider any certificate copy for an edge $e=uv$.  Since $S$ is
independent, at least one of $u$ and $v$ is not in $S$.  Suppose
$u\notin S$.  Then $o_u$ has been deleted, so the branch

\[
p_e^i\to g_{u,e}^i\to q_e^i
\]

is safe: its leakage arc reaches only the deleted selector $o_u$.

For every certificate copy, keep one safe branch and delete the other branch
gate if necessary.  These additional gate deletions have zero cost.  Hence a
complete directed path from $a$ to $b$ remains through the serial chain.

Every reachable surviving branch gate leaks only to a deleted selector.  All
unsafe branch gates have been deleted.  Therefore no vertex reachable from
$a$ can reach the drain tree and then $z$.  Hence

\[
a\not\leadsto z
\]

after the deletion, while

\[
a\leadsto b
\]

is preserved.  Thus the RPMNC instance has a feasible solution of cost at
most $B$.
\end{proof}

\begin{lemma}[Feasible RPMNC solution implies independent set]
\label{lem:rpmnc-to-is}
If the constructed directed planar RPMNC instance has a feasible node-deletion
set of cost at most $B=n-K$, then $H$ has an independent set of size at least
$K$.
\end{lemma}

\begin{proof}
Let $X$ be a feasible RPMNC node-deletion set with

\[
w(X)\le B.
\]

No vertex of weight $M=B+1$ can belong to $X$.  Thus $X$ contains only
selectors and zero-cost branch gates.

Define

\[
S=\{u\in U:o_u\notin X\}.
\]

We claim that $S$ is independent.  Suppose not.  Then there is an edge

\[
e=uv
\]

with $u,v\in S$.  Consider any certificate copy $D_e^i$ of this edge.  Since
$a\leadsto b$ must survive, every serial certificate copy must contain at
least one surviving branch from its entry to its exit.  Therefore at least
one of

\[
g_{u,e}^i,\qquad g_{v,e}^i
\]

survives.  Suppose $g_{u,e}^i$ survives.  The surviving $a$-$b$ path reaches
$g_{u,e}^i$, and since $o_u\notin X$, the leakage path

\[
g_{u,e}^i\to o_u\leadsto z
\]

also survives.  Hence

\[
a\leadsto z,
\]

contradicting feasibility.  Therefore no edge of $H$ has both endpoints in
$S$, so $S$ is independent.

Finally,

\[
|U\setminus S|
=
|\{u\in U:o_u\in X\}|
\le
w(X)
\le
n-K.
\]

Thus

\[
|S|\ge K.
\]

So $H$ has an independent set of size at least $K$.
\end{proof}

\begin{theorem}[Finite-bimodal directed planar RPMNC hardness]
\label{thm:finite-bimodal-rpmnc-hard}
Directed planar RPMNC is NP-hard even when every finite-cost deletable vertex
is bimodal.  The hardness holds even for acyclic directed planar instances,
assuming nonnegative node weights are allowed.
\end{theorem}

\begin{proof}
The reduction above is polynomial, planar, and acyclic by construction.
Lemmas~\ref{lem:is-to-rpmnc} and~\ref{lem:rpmnc-to-is} prove correctness.

It remains only to verify the finite-bimodality property.  The finite-cost
vertices are exactly the selector vertices $o_u$ and the branch gates
$g_{u,e}^i$.  Each branch gate has one incoming arc from $p_e^i$ and two
outgoing arcs, one to $q_e^i$ and one to $o_u$.  Thus its local incidence
pattern is

\[
I,O,O,
\]

which is bimodal in every cyclic order because there is only one incoming
arc.

Each selector $o_u$ has only incoming leakage arcs and one outgoing drain arc
to the drain tree.  Thus its local incidence pattern is

\[
I,I,\ldots,I,O.
\]

Since there is only one outgoing arc, the incoming arcs form one consecutive
block in the cyclic order, and the outgoing block is a singleton.  Hence every
selector is also bimodal.

Therefore every finite-cost deletable vertex is bimodal.
\end{proof}

\subsection{From Finite-Bimodal RPMNC to Directed RPMEC}

We now transfer the node-deletion hardness to edge deletion.

\begin{lemma}[Finite-bimodal node-to-edge split]
\label{lem:finite-bimodal-split}
Let $G$ be a directed planar RPMNC instance with terminals $s_1,s_2,t$,
budget $B$, and node weights $w$.  Suppose every nonterminal vertex $v$ with

\[
w(v)\le B
\]

is bimodal in the given planar embedding.  Then one can construct in
polynomial time a directed planar RPMEC instance $G'$ with the same terminals
and budget $B$ such that $G$ has a feasible node-deletion solution of cost at
most $B$ if and only if $G'$ has a feasible arc-deletion solution of cost at
most $B$.
\end{lemma}

\begin{proof}
Let

\[
D=\{v\in V(G)\setminus\{s_1,s_2,t\}:w(v)\le B\}.
\]

Only vertices in $D$ can be deleted by a feasible node solution.  For each
$v\in D$, replace $v$ by two vertices $v^{-}$ and $v^{+}$ and one gate arc

\[
g_v=(v^{-},v^{+})
\]

with arc cost

\[
c(g_v)=w(v).
\]

For every original arc $(u,v)$, replace its tail by $u^{+}$ if $u\in D$ and
by $u$ otherwise.  Replace its head by $v^{-}$ if $v\in D$ and by $v$
otherwise.  The resulting connector arc receives cost

\[
B+1.
\]

Terminals are not split.

The split preserves planarity.  Indeed, for each $v\in D$, the incoming arcs
of $v$ form one consecutive sector and the outgoing arcs form the other
sector.  Place $v^{-}$ in the incoming sector, place $v^{+}$ in the outgoing
sector, attach all incoming arcs to $v^{-}$, attach all outgoing arcs from
$v^{+}$, and draw the gate arc $v^{-}\to v^{+}$ inside the small disk that
formerly contained $v$.  Since incoming and outgoing arcs are not interleaved,
no crossing is introduced.  Vertices outside $D$ are not split.

Now suppose $X\subseteq D$ is a feasible node-deletion solution in $G$.
Delete

\[
F_X=\{g_v:v\in X\}
\]

in $G'$.  Then

\[
c(F_X)=w(X)\le B.
\]

Every directed path in $G-X$ lifts to a directed path in $G'-F_X$ by replacing
each surviving split vertex $v$ by the gate arc $v^{-}\to v^{+}$.  Therefore
$s_1\leadsto s_2$ is preserved in $G'-F_X$.  If $s_1\leadsto t$ existed in
$G'-F_X$, contracting every surviving gate arc would yield an $s_1$-$t$ path
in $G-X$, contradiction.  Hence $F_X$ is feasible for RPMEC.

Conversely, let $F$ be a feasible RPMEC arc-deletion solution in $G'$ with

\[
c(F)\le B.
\]

No connector arc can belong to $F$, because every connector arc has cost
$B+1$.  Thus $F$ consists only of split gate arcs.  Define

\[
X_F=\{v\in D:g_v\in F\}.
\]

Then

\[
w(X_F)=c(F)\le B.
\]

Since $s_1\leadsto s_2$ survives in $G'-F$, contracting every surviving split
gate gives an $s_1$-$s_2$ path in $G-X_F$.  If $s_1\leadsto t$ existed in
$G-X_F$, expanding every surviving split vertex into its gate arc would give
an $s_1$-$t$ path in $G'-F$, contradicting feasibility.  Thus $X_F$ is a
feasible RPMNC solution in $G$.
\end{proof}

\begin{theorem}[Directed planar RPMEC hardness]
\label{thm:directed-planar-rpmec-hard}
The one-way directed three-terminal RPMEC problem is NP-hard on planar
directed graphs.  The hardness holds even for acyclic planar directed graphs,
assuming nonnegative arc costs are allowed.
\end{theorem}

\begin{proof}
By Theorem~\ref{thm:finite-bimodal-rpmnc-hard}, directed planar RPMNC is
NP-hard even on acyclic instances in which every finite-cost deletable vertex
is bimodal.  Applying the finite-bimodal node-to-edge split of
Lemma~\ref{lem:finite-bimodal-split} gives an equivalent directed planar
RPMEC instance.

The split also preserves acyclicity.  Given a topological ordering of the
RPMNC instance, replace each split vertex $v$ by the ordered pair

\[
v^{-}<v^{+}
\]

at the position formerly occupied by $v$.  Every connector arc follows the
original topological order, and every split gate arc goes from $v^{-}$ to
$v^{+}$.  Hence the resulting RPMEC instance is acyclic.

Therefore directed planar RPMEC is NP-hard, even on acyclic planar directed
graphs.
\end{proof}

\begin{remark}[Strictly positive costs]
The proof above uses zero-cost branch gates.  If the problem definition
requires strictly positive rational deletion costs, let $N$ be the number of
branch gates and assign each branch gate cost

\[
\varepsilon=\frac{1}{N+1}.
\]

Replace the budget $B$ by

\[
B'=B+\frac{N}{N+1}.
\]

Deleting all branch gates costs less than $1$, while deleting one additional
selector costs $1$.  Therefore the selector-count logic of the reduction is
unchanged.

If integer positive costs are required, multiply all selector costs by
$N+1$, give every branch gate cost $1$, and replace the budget by

\[
B''=(N+1)B+N.
\]

Again, deleting all branch gates costs at most $N$, while one additional
selector costs $N+1$, so the reduction remains valid.
\end{remark}

\section{Conclusion}

We studied the one-way directed three-terminal RPMEC problem, in which one
must preserve $s_1\leadsto s_2$ while destroying $s_1\leadsto t$.  The main
algorithmic tool is a path--cut formulation: every feasible solution can be
viewed as choosing a surviving directed $s_1$-$s_2$ path and then cutting that
path from $t$ using a rooted directed cut function.  This formulation yields
an $O(\sqrt r)$ approximation in general directed graphs via root-linear
polymatroid approximation, where $r$ is the number of relevant vertices with
positive singleton cut value.  For acyclic directed graphs, a singleton-cut
path algorithm gives an additional $h$-approximation, leading to the combined
ratio $O(\min\{\sqrt r,h\})$.

On the hardness side, we proved that directed planar RPMEC is NP-hard, even
on acyclic planar digraphs with nonnegative costs.  The proof first constructs
a finite-bimodal directed planar node-deletion instance from independent set
on cubic planar graphs.  The local bimodality of all finite-cost vertices
then permits a planar node-to-edge split, transferring the hardness from
RPMNC to RPMEC.  This shows that directed planar RPMEC behaves differently
from the undirected planar edge version: directionality is sufficient to
recover NP-hardness even in acyclic planar instances.

Several questions remain open.  The approximation factor for general directed
graphs may be improvable on special graph classes beyond DAGs of small
relevant depth.  It would also be interesting to identify structural
conditions under which the directed planar problem becomes polynomial-time
solvable, and to determine whether the $O(\sqrt r)$ guarantee is tight for
directed RPMEC under standard complexity assumptions.

\bibliographystyle{plain}
\bibliography{main}

\end{document}